\documentclass[conference]{IEEEtran}
\IEEEoverridecommandlockouts
\usepackage[hmargin=0.6in,vmargin={0.5in,0.55in}]{geometry}
\usepackage{textcomp}
\usepackage{pgfplots}
\usepackage{float}
\usepackage[compact]{titlesec}
\titlespacing{\section}{10pt}{3pt plus 1pt minus 1pt}{2pt plus 1pt minus 1pt}
\titlespacing\subsection{0pt}{1pt}{1pt plus 1pt minus 1pt}
\IEEEaftertitletext{\vspace{-1\baselineskip}}
\usepackage{relsize}
\usepackage{booktabs, longtable, multirow, supertabular, tabularx}

    \newcolumntype{P}[1]{>{\centering\hspace{0pt}\arraybackslash}p{#1}}
    \newcolumntype{M}[1]{>{\centering\hspace{0pt}\arraybackslash}m{#1}}
    \newcolumntype{L}{>{\centering\arraybackslash}m{3cm}}

\usepackage{cite}
\def\BibTeX{{\rm B\kern-.05em{\sc i\kern-.025em b}\kern-.08em
    T\kern-.1667em\lower.7ex\hbox{E}\kern-.125emX}}
\usepackage{textcomp}
\usepackage{enumitem}
\setlist{nolistsep,leftmargin=.6cm}
\usepackage{bbm}
\usepackage{amsmath, amssymb, amsthm}
\usepackage{psfrag}
\usepackage{graphicx}
\usepackage{tikz}
\usetikzlibrary{arrows.meta,
                fit,
                positioning,
                quotes}

\usepackage{cuted}
\usepackage[font=small]{caption}
\usepackage{pbox}
\usepackage{filecontents}

\usepackage{algorithm}
\usepackage{algpseudocode}
\usepackage{fancyhdr}
 \newcommand{\ind}{\perp\!\!\!\!\perp}
\newtheorem{definition}{Definition}
\newtheorem{lemma}{Lemma}
\newtheorem*{theorem*}{Theorem}
\newtheorem{theorem}{Theorem}

\begin{document}
\title{Active User Identification in Fast Fading Massive Random Access Channels}

\author{\IEEEauthorblockN{Jyotish Robin, Elza Erkip}
\IEEEauthorblockA{\textit{Dept. of Electrical and Computer Engineering, New York University, Brooklyn, NY, USA}}
}
\maketitle
\begin{abstract} Reliable and prompt identification of active users is critical for enabling random access in  massive machine-to-machine type networks which typically operate within stringent access delay and energy constraints.  In this paper, an energy efficient active user identification protocol is envisioned in which the  active users simultaneously transmit  On-Off Keying (OOK) modulated preambles whereas the  base station uses non-coherent detection to avoid the channel estimation overheads. The minimum number of channel-uses required for active user identification in the asymptotic regime of total number of users $\ell$ when the number of active devices $k$ scales as $k=\Theta(1)$  is characterized along with an achievability scheme relying on the equivalence of activity detection to a group testing problem. A practical scheme for active user identification based on a belief propagation strategy is also proposed and its performance is compared against the theoretical bounds. 
\end{abstract}
\vspace{-0cm}

\begin{IEEEkeywords}
  massive random access, activity detection, active user identification, IoT, group testing.
\end{IEEEkeywords}

\section{introduction}  In contrast to the traditional cellular systems, massive machine-type communication (mMTC) networks supporting ultra-reliable low-latency applications  have to tackle  the novel challenges posed by application-specific constraints and dynamic service requirements \cite{7736615}. Typically,  mMTC networks are  comprised of sensors, actuators, etc.  monitoring critical information which leads to a sparse and sporadic data traffic since only a small subset of devices are active at any given time instant \cite{9023459}. Since the data transmissions are infrequent, energy conserving modulation schemes such as  On-Off Keying (OOK) are recommended for battery operated mMTC applications \cite{5937958}.
In addition, mMTC systems are severely prone to security threats such as jamming, spoofing, etc. as the devices have limited computational capabilities \cite{8386824}.  To combat jamming attacks, spread-spectrum techniques have been considered for commercial networks \cite{7036828}. For example,  many mMTC architectures rely on fast frequency hopping where symbols experience independent fading to provide security against jamming \cite{9136922}.  

To properly allocate resources and to decode data in mMTC systems, the BS needs to accurately identify the active devices as misdetections can lead to critical data loss \cite{ALTURJMAN2017299}. In  this paper, we focus on the active device identification  problem in  mMTC systems with stringent requirements in terms of latency, reliability, energy efficiency, and
security. In dense M2M networks, grant-based random access (RA)  results in excessive preamble collisions drastically increasing the signaling overheads whereas grant-free schemes suffer from severe co-channel interference due to the non-orthogonality of preambles \cite{9060999}. Thus, active device identification in massive mMTC systems requires novel strategies.

 Several recent works \cite{8734871,7952810,7282735}  exploit the sporadic device activity to employ compressed sensing techniques such as Approximate Message Passing (AMP) \cite{8734871,7952810} and sparse graph codes \cite{7282735} to detect active devices. In AMP based approaches, preamble sequences are typically generated from i.i.d. Gaussian distributions which requires devices to transmit during all channel-uses consuming more energy in comparison to OOK signaling schemes \cite{5937958}. In \cite{8262800}, an RA protocol for short packet transmission was presented by viewing activity detection as a  group testing problem. Our previous works in \cite{9500808,9593173} address the activity detection problem  taking into account the  access delay and energy constraints. However, most of the existing literature assumes a block-fading channel model where the channel is static in each block in which channel state information (CSI) can be accurately obtained at the BS through pilots before device identification. For secure mMTC networks such as the fast frequency hopping systems, a more accurate channel model would incorporate fast fading where reliable channel estimates cannot be obtained \cite{923716}. Lack of reliable CSI leads to poor decoding performance with coherent detection methods and hence  alternatives that do not require channel estimation are often more efficient \cite{7514754}.

In  this paper, unlike the coherent detection approaches, we focus on identifying the  devices using a threshold-based energy detection which is  a low-complex energy efficient  non-coherent method \cite{5174497}. Active devices transmit their OOK modulated unique preambles synchronously and  BS produces a binary output  based on the received energy without requiring precise CSI. Our aim is to estimate the set of active devices reliably using the  binary energy detector outputs.

We argue that the above active device identification framework  can be viewed as an instance of the group testing (GT) problem 
\cite{10.2307/2284447, 6157065} where the aim is to identify a small set of ``defective''
items within a large population by performing tests on pools of items. A test
is positive if the pool contains at least one defective, and negative otherwise. Though connection between GT and RA  has previously been studied \cite{1096146,1057026}, the channel models used are often unrealistic for many mMTC systems such as the non-coherent ones described above. In our paper, we bridge this gap by using a non-coherent GT model. The main contributions of this paper are as follows:
\begin{itemize}
\item 
We propose a novel non-coherent activity detection framework  for  mMTC networks in which active devices jointly transmit their OOK preambles  and  the BS uses a non-coherent energy detection to identify the active devices.

\item We use  information theoretic  tools to characterize the minimum user identification cost defined as the minimum number of channel-uses required to identify the active users \cite{robin2021capacity}.  Using GT, we derive an achievable minimum user identification cost for our framework in the $k=\Theta(1)$ regime where $k$ is the number of active devices.

\item We present a  practical non-coherent device identification scheme based on Belief Propagation (BP) which employs bitwise-Maximum Aposteriori Probability (MAP) detection \cite{sejdinovic2010note}. We evaluate its  performance in terms of user identification cost w.r.t our theoretical results.
   \end{itemize}

\section{System model}~\label{sec:sysmodl} 
Consider an mMTC network consisting of $\ell$ users belonging to the set $\mathcal{D}=\{1,2,\ldots,\ell\}$ out of which only $k$  are active and wants to access the channel. We assume that the value of $k$ is estimated at the BS \textit{a-priori} \cite{4085381}. The active user set  denoted by $\mathcal{A}=\{a_1,a_2\ldots, a_k\}$  is assumed to be randomly chosen among all ${\ell \choose k} $ subsets of size $k$ from $\mathcal{D}$ and is unknown to the BS. The BS aims to identify the set $\mathcal{A}$ in a time-efficient manner. 
\begin{figure}   
	\centering
	\resizebox{3.4in}{!}{
\tikzset{every picture/.style={line width=0.95pt}} 

\begin{tikzpicture}[x=1pt,y=0.99pt,yscale=-.9,xscale=.99]

\draw  [dash pattern={on 0.84pt off 2.51pt}]  (55.67,152.17) -- (56.08,193.14) -- (56.67,251.17) ;
\draw    (153,62) -- (199.67,62.16) ;
\draw [shift={(201.67,62.17)}, rotate = 180.2] [color={rgb, 255:red, 0; green, 0; blue, 0 }  ][line width=0.75]    (10.93,-3.29) .. controls (6.95,-1.4) and (3.31,-0.3) .. (0,0) .. controls (3.31,0.3) and (6.95,1.4) .. (10.93,3.29)   ;
\draw    (151,137) -- (197.67,137.16) ;
\draw [shift={(199.67,137.17)}, rotate = 180.2] [color={rgb, 255:red, 0; green, 0; blue, 0 }  ][line width=0.75]    (10.93,-3.29) .. controls (6.95,-1.4) and (3.31,-0.3) .. (0,0) .. controls (3.31,0.3) and (6.95,1.4) .. (10.93,3.29)   ;
\draw    (151,268) -- (197.67,268.16) ;
\draw [shift={(199.67,268.17)}, rotate = 180.2] [color={rgb, 255:red, 0; green, 0; blue, 0 }  ][line width=0.75]    (10.93,-3.29) .. controls (6.95,-1.4) and (3.31,-0.3) .. (0,0) .. controls (3.31,0.3) and (6.95,1.4) .. (10.93,3.29)   ;
\draw   (225.72,51.59) .. controls (231.16,57.15) and (230.77,66.35) .. (224.85,72.15) .. controls (218.93,77.94) and (209.72,78.14) .. (204.28,72.58) .. controls (198.84,67.02) and (199.23,57.82) .. (205.15,52.02) .. controls (211.07,46.22) and (220.28,46.03) .. (225.72,51.59) -- cycle ; \draw   (225.72,51.59) -- (204.28,72.58) ; \draw   (224.85,72.15) -- (205.15,52.02) ;
\draw   (227.72,125.59) .. controls (233.16,131.15) and (232.77,140.35) .. (226.85,146.15) .. controls (220.93,151.94) and (211.72,152.14) .. (206.28,146.58) .. controls (200.84,141.02) and (201.23,131.82) .. (207.15,126.02) .. controls (213.07,120.22) and (222.28,120.03) .. (227.72,125.59) -- cycle ; \draw   (227.72,125.59) -- (206.28,146.58) ; \draw   (226.85,146.15) -- (207.15,126.02) ;
\draw   (225.72,259.59) .. controls (231.16,265.15) and (230.77,274.35) .. (224.85,280.15) .. controls (218.93,285.94) and (209.72,286.14) .. (204.28,280.58) .. controls (198.84,275.02) and (199.23,265.82) .. (205.15,260.02) .. controls (211.07,254.22) and (220.28,254.03) .. (225.72,259.59) -- cycle ; \draw   (225.72,259.59) -- (204.28,280.58) ; \draw   (224.85,280.15) -- (205.15,260.02) ;
\draw    (214.67,233.17) -- (214.06,252) ;
\draw [shift={(214,254)}, rotate = 271.83] [color={rgb, 255:red, 0; green, 0; blue, 0 }  ][line width=0.75]    (10.93,-3.29) .. controls (6.95,-1.4) and (3.31,-0.3) .. (0,0) .. controls (3.31,0.3) and (6.95,1.4) .. (10.93,3.29)   ;
\draw    (215.67,103.17) -- (215.96,120) ;
\draw [shift={(216,122)}, rotate = 268.99] [color={rgb, 255:red, 0; green, 0; blue, 0 }  ][line width=0.75]    (10.93,-3.29) .. controls (6.95,-1.4) and (3.31,-0.3) .. (0,0) .. controls (3.31,0.3) and (6.95,1.4) .. (10.93,3.29)   ;
\draw    (214.67,26.17) -- (214.96,43) ;
\draw [shift={(215,45)}, rotate = 268.99] [color={rgb, 255:red, 0; green, 0; blue, 0 }  ][line width=0.75]    (10.93,-3.29) .. controls (6.95,-1.4) and (3.31,-0.3) .. (0,0) .. controls (3.31,0.3) and (6.95,1.4) .. (10.93,3.29)   ;
\draw    (228,69) -- (351.06,159.98) ;
\draw [shift={(352.67,161.17)}, rotate = 216.48] [color={rgb, 255:red, 0; green, 0; blue, 0 }  ][line width=0.75]    (10.93,-3.29) .. controls (6.95,-1.4) and (3.31,-0.3) .. (0,0) .. controls (3.31,0.3) and (6.95,1.4) .. (10.93,3.29)   ;
\draw    (232.67,142.17) -- (350.71,167.74) ;
\draw [shift={(352.67,168.17)}, rotate = 192.23] [color={rgb, 255:red, 0; green, 0; blue, 0 }  ][line width=0.75]    (10.93,-3.29) .. controls (6.95,-1.4) and (3.31,-0.3) .. (0,0) .. controls (3.31,0.3) and (6.95,1.4) .. (10.93,3.29)   ;
\draw    (229.67,269.17) -- (353.07,176.37) ;
\draw [shift={(354.67,175.17)}, rotate = 143.06] [color={rgb, 255:red, 0; green, 0; blue, 0 }  ][line width=0.75]    (10.93,-3.29) .. controls (6.95,-1.4) and (3.31,-0.3) .. (0,0) .. controls (3.31,0.3) and (6.95,1.4) .. (10.93,3.29)   ;
\draw  [dash pattern={on 0.84pt off 2.51pt}]  (213.67,155.17) -- (214.67,233.17) ;
\draw   (352.67,168.17) .. controls (352.67,154.36) and (363.86,143.17) .. (377.67,143.17) .. controls (391.47,143.17) and (402.67,154.36) .. (402.67,168.17) .. controls (402.67,181.97) and (391.47,193.17) .. (377.67,193.17) .. controls (363.86,193.17) and (352.67,181.97) .. (352.67,168.17) -- cycle ;
\draw    (376.67,103.17) -- (377.62,141.17) ;
\draw [shift={(377.67,143.17)}, rotate = 268.57] [color={rgb, 255:red, 0; green, 0; blue, 0 }  ][line width=0.75]    (10.93,-3.29) .. controls (6.95,-1.4) and (3.31,-0.3) .. (0,0) .. controls (3.31,0.3) and (6.95,1.4) .. (10.93,3.29)   ;
\draw    (402.67,168.17) -- (443.67,168.17) ;
\draw [shift={(445.67,168.17)}, rotate = 180] [color={rgb, 255:red, 0; green, 0; blue, 0 }  ][line width=0.75]    (10.93,-3.29) .. controls (6.95,-1.4) and (3.31,-0.3) .. (0,0) .. controls (3.31,0.3) and (6.95,1.4) .. (10.93,3.29)   ;
\draw   (447,152.33) .. controls (447,146.17) and (452,141.17) .. (458.17,141.17) -- (505.83,141.17) .. controls (512,141.17) and (517,146.17) .. (517,152.33) -- (517,185.83) .. controls (517,192) and (512,197) .. (505.83,197) -- (458.17,197) .. controls (452,197) and (447,192) .. (447,185.83) -- cycle ;
\draw    (517.67,167.17) -- (556.67,167.17) ;
\draw [shift={(558.67,167.17)}, rotate = 180] [color={rgb, 255:red, 0; green, 0; blue, 0 }  ][line width=0.75]    (10.93,-3.29) .. controls (6.95,-1.4) and (3.31,-0.3) .. (0,0) .. controls (3.31,0.3) and (6.95,1.4) .. (10.93,3.29)   ;
\draw   (559,151.33) .. controls (559,145.17) and (564,140.17) .. (570.17,140.17) -- (640.5,140.17) .. controls (646.67,140.17) and (651.67,145.17) .. (651.67,151.33) -- (651.67,184.83) .. controls (651.67,191) and (646.67,196) .. (640.5,196) -- (570.17,196) .. controls (564,196) and (559,191) .. (559,184.83) -- cycle ;
\draw    (651.67,167.17) -- (668.67,168.06) ;
\draw [shift={(670.67,168.17)}, rotate = 183.01] [color={rgb, 255:red, 0; green, 0; blue, 0 }  ][line width=0.75]    (10.93,-3.29) .. controls (6.95,-1.4) and (3.31,-0.3) .. (0,0) .. controls (3.31,0.3) and (6.95,1.4) .. (10.93,3.29)   ;
\draw    (68,62) -- (112,62) ;
\draw    (112,62) -- (149,39) ;

\draw    (71,136) -- (115,136) ;
\draw    (115,136) -- (152,113) ;

\draw    (70,267) -- (114,267) ;
\draw    (114,267) -- (151,244) ;

\draw (44,48.4) node [anchor=north west][inner sep=0.75pt]   [font=\LARGE]{$X_{t}^{1}$};
\draw (47,120.4) node [anchor=north west][inner sep=0.75pt]   [font=\LARGE]{$X_{t}^{2}$};
\draw (46,256.4) node [anchor=north west][inner sep=0.75pt]   [font=\LARGE] {$X_{t}^{\ell}$};
\draw (219,18.4) node [anchor=north west][inner sep=0.75pt]  [font=\LARGE] {$h_t^{1}$};
\draw (220,95.4) node [anchor=north west][inner sep=0.75pt]  [font=\LARGE]  {$h_t^{2}$};
\draw (217,222.4) node [anchor=north west][inner sep=0.75pt] [font=\LARGE]  {$h_t^{\ell}$};
\draw (363,153.4) node [anchor=north west][inner sep=0.75pt]  [font=\huge]  {$\sum $};
\draw (367,80.4) node [anchor=north west][inner sep=0.75pt]  [font=\LARGE] {$w_{t}$};
\draw (414,143.4) node [anchor=north west][inner sep=0.75pt]  [font=\LARGE]  {$S_{t}$};
\draw (559,159) node [anchor=north west][inner sep=0.75pt]  [font=\LARGE][align=left] {Thresholding};
\draw (450,149.33) node [anchor=north west][inner sep=0.75pt]  [font=\LARGE] [align=left] {Envelope \\ detector};
\draw (527,145.4) node [anchor=north west][inner sep=0.75pt]  [font=\LARGE] {$Y_{t}$};
\draw (675,158.4) node [anchor=north west][inner sep=0.75pt][font=\LARGE]  {$Z_{t}$};
\draw (124,54.4) node [anchor=north west][inner sep=0.75pt]  [font=\LARGE]{$\beta _{1}$};
\draw (127,128.4) node [anchor=north west][inner sep=0.75pt]   [font=\LARGE]{$\beta _{2}$};
\draw (126,259.4) node [anchor=north west][inner sep=0.75pt] [font=\LARGE]  {$\beta _{\ell}$};

\end{tikzpicture}

}
	\setlength{\belowcaptionskip}{-7pt}
	\caption{Non-coherent $(\ell,k)$-Many Access Channel. }
	\label{fig:qtzn}
	\end{figure}
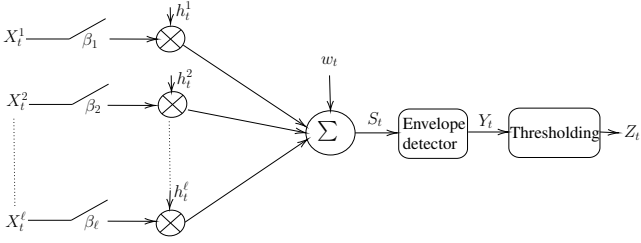
	
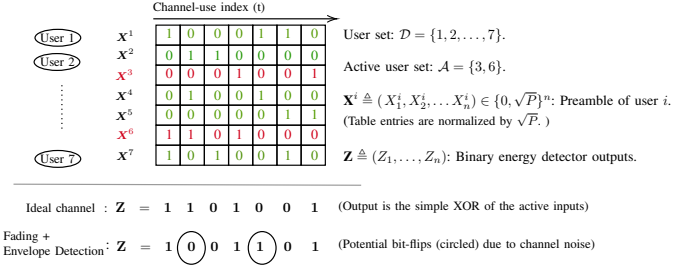
\begin{figure}   
	\centering
	\resizebox{3.6in}{!}{
\tikzset{every picture/.style={line width=0.85pt}} 

\begin{tikzpicture}[x=0.75pt,y=0.75pt,yscale=-1,xscale=.85]

\draw  [dash pattern={on 0.84pt off 2.51pt}]  (59,90) -- (59,132) ;
\draw   (28,68) .. controls (28,63.58) and (39.42,60) .. (53.5,60) .. controls (67.58,60) and (79,63.58) .. (79,68) .. controls (79,72.42) and (67.58,76) .. (53.5,76) .. controls (39.42,76) and (28,72.42) .. (28,68) -- cycle ;
\draw   (29,45) .. controls (29,40.58) and (40.42,37) .. (54.5,37) .. controls (68.58,37) and (80,40.58) .. (80,45) .. controls (80,49.42) and (68.58,53) .. (54.5,53) .. controls (40.42,53) and (29,49.42) .. (29,45) -- cycle ;

\draw   (29,160) .. controls (29,155.58) and (40.42,152) .. (54.5,152) .. controls (68.58,152) and (80,155.58) .. (80,160) .. controls (80,164.42) and (68.58,168) .. (54.5,168) .. controls (40.42,168) and (29,164.42) .. (29,160) -- cycle ;

\draw   (164,33) -- (191,33) -- (191,52) -- (164,52) -- cycle ;
\draw   (191,33) -- (218,33) -- (218,52) -- (191,52) -- cycle ;
\draw   (218,33) -- (245,33) -- (245,52) -- (218,52) -- cycle ;
\draw   (245,33) -- (272,33) -- (272,52) -- (245,52) -- cycle ;
\draw   (272,33) -- (299,33) -- (299,52) -- (272,52) -- cycle ;
\draw   (299,33) -- (326,33) -- (326,52) -- (299,52) -- cycle ;
\draw   (326,33) -- (353,33) -- (353,52) -- (326,52) -- cycle ;

\draw   (164,52) -- (191,52) -- (191,71) -- (164,71) -- cycle ;
\draw   (191,52) -- (218,52) -- (218,71) -- (191,71) -- cycle ;
\draw   (218,52) -- (245,52) -- (245,71) -- (218,71) -- cycle ;
\draw   (245,52) -- (272,52) -- (272,71) -- (245,71) -- cycle ;
\draw   (272,52) -- (299,52) -- (299,71) -- (272,71) -- cycle ;
\draw   (299,52) -- (326,52) -- (326,71) -- (299,71) -- cycle ;
\draw   (326,52) -- (353,52) -- (353,71) -- (326,71) -- cycle ;

\draw   (164,71) -- (191,71) -- (191,90) -- (164,90) -- cycle ;
\draw   (191,71) -- (218,71) -- (218,90) -- (191,90) -- cycle ;
\draw   (218,71) -- (245,71) -- (245,90) -- (218,90) -- cycle ;
\draw   (245,71) -- (272,71) -- (272,90) -- (245,90) -- cycle ;
\draw   (272,71) -- (299,71) -- (299,90) -- (272,90) -- cycle ;
\draw   (299,71) -- (326,71) -- (326,90) -- (299,90) -- cycle ;
\draw   (326,71) -- (353,71) -- (353,90) -- (326,90) -- cycle ;

\draw   (164,147) -- (191,147) -- (191,166) -- (164,166) -- cycle ;
\draw   (191,147) -- (218,147) -- (218,166) -- (191,166) -- cycle ;
\draw   (218,147) -- (245,147) -- (245,166) -- (218,166) -- cycle ;
\draw   (245,147) -- (272,147) -- (272,166) -- (245,166) -- cycle ;
\draw   (272,147) -- (299,147) -- (299,166) -- (272,166) -- cycle ;
\draw   (299,147) -- (326,147) -- (326,166) -- (299,166) -- cycle ;
\draw   (326,147) -- (353,147) -- (353,166) -- (326,166) -- cycle ;

\draw   (164,90) -- (191,90) -- (191,109) -- (164,109) -- cycle ;
\draw   (191,90) -- (218,90) -- (218,109) -- (191,109) -- cycle ;
\draw   (218,90) -- (245,90) -- (245,109) -- (218,109) -- cycle ;
\draw   (245,90) -- (272,90) -- (272,109) -- (245,109) -- cycle ;
\draw   (272,90) -- (299,90) -- (299,109) -- (272,109) -- cycle ;
\draw   (299,90) -- (326,90) -- (326,109) -- (299,109) -- cycle ;
\draw   (326,90) -- (353,90) -- (353,109) -- (326,109) -- cycle ;

\draw   (164,109) -- (191,109) -- (191,128) -- (164,128) -- cycle ;
\draw   (191,109) -- (218,109) -- (218,128) -- (191,128) -- cycle ;
\draw   (218,109) -- (245,109) -- (245,128) -- (218,128) -- cycle ;
\draw   (245,109) -- (272,109) -- (272,128) -- (245,128) -- cycle ;
\draw   (272,109) -- (299,109) -- (299,128) -- (272,128) -- cycle ;
\draw   (299,109) -- (326,109) -- (326,128) -- (299,128) -- cycle ;
\draw   (326,109) -- (353,109) -- (353,128) -- (326,128) -- cycle ;

\draw   (164,128) -- (191,128) -- (191,147) -- (164,147) -- cycle ;
\draw   (191,128) -- (218,128) -- (218,147) -- (191,147) -- cycle ;
\draw   (218,128) -- (245,128) -- (245,147) -- (218,147) -- cycle ;
\draw   (245,128) -- (272,128) -- (272,147) -- (245,147) -- cycle ;
\draw   (272,128) -- (299,128) -- (299,147) -- (272,147) -- cycle ;
\draw   (299,128) -- (326,128) -- (326,147) -- (299,147) -- cycle ;
\draw   (326,128) -- (353,128) -- (353,147) -- (326,147) -- cycle ;

\draw    (161,27) -- (364,27) ;
\draw [shift={(366,27)}, rotate = 180] [color={rgb, 255:red, 0; green, 0; blue, 0 }  ][line width=0.75]    (10.93,-3.29) .. controls (6.95,-1.4) and (3.31,-0.3) .. (0,0) .. controls (3.31,0.3) and (6.95,1.4) .. (10.93,3.29)   ;
\draw [color={rgb, 255:red, 155; green, 155; blue, 155 }  ,draw opacity=1 ]   (5,184) -- (393,185) ;
\draw   (188,244.5) .. controls (188,235.94) and (194.94,229) .. (203.5,229) .. controls (212.06,229) and (219,235.94) .. (219,244.5) .. controls (219,253.06) and (212.06,260) .. (203.5,260) .. controls (194.94,260) and (188,253.06) .. (188,244.5) -- cycle ;
\draw   (267,244.5) .. controls (267,235.94) and (273.94,229) .. (282.5,229) .. controls (291.06,229) and (298,235.94) .. (298,244.5) .. controls (298,253.06) and (291.06,260) .. (282.5,260) .. controls (273.94,260) and (267,253.06) .. (267,244.5) -- cycle ;

\draw (173,35) node [anchor=north west][inner sep=0.75pt]  [color={rgb, 255:red, 87; green, 155; blue, 7 }  ,opacity=1 ] [align=left] {1 \ \ 0 \ \ \ 0 \ \ 0 \ \ 1 \ \ \ 1 \ \ \ 0 \ };
\draw (173,55) node [anchor=north west][inner sep=0.75pt]  [color={rgb, 255:red, 36; green, 151; blue, 6 }  ,opacity=1 ] [align=left] {0 \ \ 1 \ \ \ 1 \ \ 0 \ \ 0 \ \ \ 0 \ \ \ 0 \ };
\draw (173,73) node [anchor=north west][inner sep=0.75pt]  [color={rgb, 255:red, 208; green, 2; blue, 27 }  ,opacity=1 ] [align=left] {0 \ \ 0 \ \ \ 0 \ \ 1 \ \ 0 \ \ \ 0 \ \ \ 1 \ };
\draw (173,93) node [anchor=north west][inner sep=0.75pt]  [color={rgb, 255:red, 79; green, 153; blue, 0 }  ,opacity=1 ] [align=left] {0 \ \ 1 \ \ \ 0 \ \ 0 \ \ 1 \ \ \ 0 \ \ \ 0 \ };
\draw (173,111) node [anchor=north west][inner sep=0.75pt]  [color={rgb, 255:red, 84; green, 156; blue, 0 }  ,opacity=1 ] [align=left] {0 \ \ 0 \ \ \ 0 \ \ 0  \ \ 0 \ \ \ 1 \ \ \ 1 \ };
\draw (173,150) node [anchor=north west][inner sep=0.75pt]  [color={rgb, 255:red, 81; green, 146; blue, 11 }  ,opacity=1 ] [align=left] {1 \ \ 0 \ \ \ 1 \ \ 0 \ \ 0 \ \ \ 1 \ \ \ 0 \ };
\draw (173,130) node [anchor=north west][inner sep=0.75pt]  [color={rgb, 255:red, 208; green, 2; blue, 27 }  ,opacity=1 ] [align=left] {1 \ \ 1 \ \ \ 0 \ \ 1 \ \ 0 \ \ \ 0 \ \ \ 0 \ };
\draw (33,154) node [anchor=north west][inner sep=0.75pt]  [font=\small] [align=left] {User $\displaystyle 7$};
\draw (33,39) node [anchor=north west][inner sep=0.75pt]  [font=\small] [align=left] {User 1};
\draw (33,62) node [anchor=north west][inner sep=0.75pt]  [font=\small] [align=left] {User 2};
\draw (118,37) node [anchor=north west][inner sep=0.75pt]  [font=\footnotesize] [align=left] {$\displaystyle \boldsymbol{X}^{1}$};
\draw (118,54) node [anchor=north west][inner sep=0.75pt]  [font=\footnotesize] [align=left] {$\displaystyle \boldsymbol{X}^{2}$};
\draw (117,74) node [anchor=north west][inner sep=0.75pt]  [font=\footnotesize] [align=left] {$\displaystyle \textcolor[rgb]{0.82,0.01,0.11}{\boldsymbol{X}^{3}}$};
\draw (118,130) node [anchor=north west][inner sep=0.75pt]  [font=\footnotesize] [align=left] {$\displaystyle \textcolor[rgb]{0.82,0.01,0.11}{\boldsymbol{X}^{6}}$};
\draw (117,112) node [anchor=north west][inner sep=0.75pt]  [font=\footnotesize] [align=left] {$\displaystyle \boldsymbol{X}^{5}$};
\draw (117,93) node [anchor=north west][inner sep=0.75pt]  [font=\footnotesize] [align=left] {$\displaystyle \boldsymbol{X}^{4}$};
\draw (118,150) node [anchor=north west][inner sep=0.75pt]  [font=\footnotesize] [align=left] {$\displaystyle \boldsymbol{X}^{7}$};
\draw (159,10) node [anchor=north west][inner sep=0.75pt]  [font=\small] [align=left] {Channel-use index (t)};
\draw (366,199) node [anchor=north west][inner sep=0.75pt]   [align=left] {{\small (Output is the simple XOR of the active inputs)}};
\draw (118,236.4) node [anchor=north west][inner sep=0.75pt]    {$\mathbf{Z\ \ =\ \ 1  \ \ \ 0\ \ \ 0\ \ \ 1\ \ \ 1 \ \ \ \ 0\ \ \ \ 1} \ $};
\draw (-8,227) node [anchor=north west][inner sep=0.75pt]   [align=left] { {\small Fading +}\\[-1pt]{\small Envelope Detection }};
\draw (117,199.4) node [anchor=north west][inner sep=0.75pt]    {$\mathbf{Z\ \ =\ \ 1\ \ \ 1\ \ \ 0\ \ \ 1\ \ \ 0\ \ \ \ 0\ \ \ \ 1} \ $};
\draw (-10,201) node [anchor=north west][inner sep=0.75pt]   [align=left] { \ \ \ \ \  {\small Ideal channel  \  :} };
\draw (106,238) node [anchor=north west][inner sep=0.75pt]   [align=left] {:};
\draw (366,235) node [anchor=north west][inner sep=0.75pt]   [align=left] {{\small (Potential bit-flips (circled) due to channel noise)}};
\draw (366,35) node [anchor=north west][inner sep=0.75pt]   [align=left] {{ User set: $\mathcal{D}=\{1,2,\ldots,7\}$.}};

\draw (366,65) node [anchor=north west][inner sep=0.75pt]   [align=left] {{ Active user set: $\mathcal{A}=\{3,6\}$.}};

\draw (366,95) node [anchor=north west][inner sep=0.75pt]   [align=left] {{ $\textbf{X}^{i} \triangleq (\hspace{0.05cm}X_1^i,X_2^i,\ldots X_{n}^i) \in \{0,\sqrt{P}\}^{n} $: Preamble of user $i$.}};
\draw (366,115) node [anchor=north west][inner sep=0.75pt]   [align=left] {\ {(\small Table entries are normalized by $\sqrt{P}$. )} };
\draw (366,148) node [anchor=north west][inner sep=0.75pt]   [align=left] {{ $\textbf{Z} \triangleq (Z_1,\ldots, Z_n)$: Binary   energy detector outputs. }};
\end{tikzpicture}
}
	\setlength\abovecaptionskip{-2pt}
	\setlength{\belowcaptionskip}{-15pt}
	\caption{GT equivalence of Non-coherent $(7,2)$-MnAC. Users 3 and 6 are assumed to be active. }
	\label{fig:qtzn2}
	\end{figure}
Taking into account sparse activity, we consider the asymptotic regime where $\ell$ can be unbounded  while the number of active users scale as $k=\Theta(\ell^{\alpha});$  $0 \leq \alpha <1$. This is  consistent with the notion of  Many Access channel (MnAC) model introduced by Chen \textit{et al.} \cite{7852531}  with the difference that instead of our setting with $k$ active users, \cite{7852531}  considers the setting where users are active  probabilistically. In this paper, we focus on the $k= \Theta(1)$ regime. We use $(\ell,k)$-MnAC  to refer to the channel model used in this paper.

For purposes of activity detection, each user $i \in \mathcal{D}$ is assigned  an i.i.d  binary preamble  $\textbf{X}^{i} \triangleq (\hspace{0.05cm}X_1^i,X_2^i,\ldots X_{n}^i) \in \{0,\sqrt{P}\}^{n} $ of length $n$ generated via a  Bernoulli process with probability  $q$. We use $\mathcal{P}$ to denote the entire set of preambles. During activity detection,  the active users transmit their OOK-modulated preamble in a time-synchronized manner over $n$ channel-uses with  power $P$ during the `On' symbols. The received symbol $S_t$ during the $t^{th}$ channel-use is 
\begin{equation}
    S_t=(\boldsymbol{\beta} \circ \mathbf{h}_t) \cdot \mathbf{X}_{t}+w_t, \forall t \in\{1,2,\ldots,n\}.
\end{equation}Here, $\textbf{X}_t=(X^1_t,\ldots, X^\ell_t)$ represents the $t^{th}$ preamble entries of the $\ell$ users and $\textbf{h}_t=(h^1_t,\ldots,h^\ell_{t})$ is the vector of channel coefficients of the $\ell$ users during the $t^{th}$ channel-use  where $h_t^i$'s are   i.i.d   The AWGN noise vector $\textbf{w}=(w_1,\ldots, w_n)$ is composed of i.i.d entries,  $w_t \sim \mathcal{C} \mathcal{N}\left(0, \sigma_w^{2}\right)$ and is independent of $\textbf{h}_t$ $\forall t\in \{1,2,\ldots,n\}$.  The vector $\boldsymbol{\beta}=(\beta_{1},\ldots,\beta_{\ell})$ represents the activity status such  that $\beta_{i}=1$, if $i^{th}$ user is active; 0 otherwise.  $\boldsymbol{\beta} \circ \mathbf{h}_t$ denotes the element-wise product between the vectors $\boldsymbol{\beta}$ and $\mathbf{h}_t$. We define the signal-to-noise ratio as $\text{SNR}:=\frac{P\sigma^2}{\sigma_w^2}.$

We assume that there is no CSI available at the BS since we are operating in a fast fading scenario such as fast frequency hopping networks \cite{6477839}. This means $\textbf{h}_t, \forall t\in \{1,2,\ldots,n\}$ is unknown \textit{a-priori} at BS. Hence, the BS employs non-coherent detection as shown in Fig. 1. During the $t^{th}$ channel-use, the envelope detector processes  $S_t$ to obtain $Y_{t}=|S_t|$, the envelope of $S_t$.  Thresholding  after envelope detection produces a binary output $Z_{t}$ such that
\begin{equation}
  Z_{t}=1 \text { if } |S_t|^2>\gamma ; \text { else } Z_{t}=0.  
  \label{threshold}
\end{equation} Here, $\gamma $ represents a pre-determined threshold parameter.

\subsection{GT model for active device identification}

Group testing \cite{6157065} considers a sparse inference problem where the objective is to identify a sparse set of ``defective" items from a large set of items by performing tests on groups of items. In the error-free scenario, each group test produces a binary output where a `1' corresponds to the inclusion of at least one ``defective" in the group being tested whereas a `0' indicates that the tested group of items excludes all defectives. The tests need to be devised such that the defective set of items can be recovered using the binary vector of test outcomes with a minimum number of tests. GT models  can be classified as  adaptive and non-adaptive models.  In adaptive GT, the previous test results can be used to design the future tests. In non-adaptive setting, all group tests are designed independent of each other.

To draw the equivalence between active device identification and GT, we can consider the devices as items and the active devices as defectives. Each channel-use $t \in \{1,2,\ldots, n\}$ corresponds to a group test in which a randomly chosen subset of devices engage in joint transmission of `On' symbols, if active. The random selection of devices during the $n$ testing instances is based on the Bernoulli based independent binary preambles  $\textbf{X}^{i}  \in \{0,\sqrt{P}\}^{n}, \forall i \in \mathcal{D}$. The BS measures received energy during each channel-use and produces a binary 0-1 output based on whether the received energy exceeds a predetermined threshold. These 1-bit energy measurements at the receiver side corresponds to the noisy group test results. This is illustrated in Fig. \ref{fig:qtzn} and Fig. \ref{fig:qtzn2}. Our framework corresponds to non-adaptive GT since the testing/grouping pattern is based on predefined binary preambles and is not updated based on the previous tests.

\subsection{Problem Formulation}
Given the preamble set $\mathcal{P}$  and the vector $\textbf{Z}=(Z_1,\ldots, Z_n)$  of channel outputs, BS aims to identify the  active user set $\mathcal{A}$, alternatively the vector $\boldsymbol{\beta}$, during the activity detection phase. We use the following definitions:

\begin{definition}
\textbf{ Activity detection}: For a given set of preambles $\{\textbf{X}^i: i\in \mathcal{D}\}$, an activity detection function  $\hat{\boldsymbol{\beta}}:\{0,1\}^{n} \rightarrow \{0,1\}^{\ell}$  is a deterministic rule that maps the binary valued channel outputs $\textbf{Z}$ to an estimate $\hat{\boldsymbol{\beta}}$ of the  activity status vector such that the Hamming weight of $\hat{\boldsymbol{\beta}}$ is $k$.
\end{definition}

\begin{definition}
\textbf{Probability of erroneous identification}: For an activity detection function  $\hat{\boldsymbol{\beta}}$, probability of erroneous identification  $\mathbb{P}_e^{(\ell)}$ is defined as \begin{equation}
     \mathbb{P}_e^{(\ell)}:=\frac{1}{{\ell \choose k}} \sum_{\boldsymbol{\beta}:\sum_{i=1}^{\ell}{\beta_i}=k} \mathbb{P}(\hat{\boldsymbol{\beta}} \neq \boldsymbol{\beta}).
 \end{equation} 
\end{definition}

\begin{definition}
\textbf{Minimum user identification cost} \cite{7852531}: The minimum user identification
cost is said to be $n(\ell)$ if for every $0 < \epsilon <1$, there exists a preamble of length $n_0 = (1 + \epsilon)n(\ell)$ such that $\mathbb{P}_e^{(\ell)} \rightarrow 0$ as $\ell \rightarrow \infty$, whereas
$\mathbb{P}_e^{(\ell)}$ is strictly bounded away from 0 if
$n_0 = (1 - \epsilon)n(\ell)$.
\end{definition}

Note that  the minimum user identification cost  $n(\ell)$ is defined based on its asymptotic behaviour as $\ell \rightarrow \infty$. Our goal is to characterize $n(\ell)$  of the non-coherent $(\ell,k)-$MnAC  in Fig. 1 which would be indicative of the time-efficiency of BS in identifying the active users. 
 
 The key distinction to the work of Chen \textit{et al}. in \cite{7852531} is that we use OOK signaling at the user side and non-coherent detection at the BS which suits secure mMTC networks where obtaining reliable CSI is infeasible due to the fast fading nature \cite{s18041291}. Furthermore, in contrast to our prior work \cite{robin2021capacity} which evaluates the performance of N-COMP based detection in the non-coherent setting, our focus here is to derive fundamental limits on user identification cost. 

	\section{minimum user identification cost}

 In this section, we derive the minimum user identification cost for the non-coherent $(\ell,k)-$ MnAC in the $k=\Theta(1)$ regime as presented in Theorem \ref{thm3}. First, we provide an equivalent characterization of the active device identification problem in  the non-coherent $(\ell,k)-$ MnAC by viewing it as a  decoding problem in a point-to-point  vector input - scalar  output  channel whose inputs correspond to the active users as shown in Fig. \ref{fig:redalpha}. Thereafter, we derive the maximum rate of the equivalent channel by exploiting its cascade structure. Then, we use  GT theory \cite{6157065} to relate the maximum rate with the minimum user identification cost of non-coherent $(\ell,k)-$MnAC.  Here onwards, we drop the subscript $t$ denoting channel-use index. All logarithms are in base 2.
\vspace{0.05cm}
 \subsection{Equivalent  Channel Model} 
 \vspace{0.05cm}
The aim of active device identification  in non-coherent $(\ell,k)-$ MnAC is to identify the set of active devices based on the binary detector output vector $\textbf{Z}$ and the preamble set $\mathcal{P}$.  This can be alternatively viewed as a decoding problem in a traditional point to point communication channel as follows:  The source message is the activity status vector $\boldsymbol{\beta}$ denoting the set of active users $\mathcal{A}$ and the codebook is the entire set of preambles $\mathcal{P}$. Given $\boldsymbol{\beta}$ for an active set $\mathcal{A}=\{a_1,\ldots a_k\}$, the codeword becomes the corresponding subset of preambles $\{\textbf{X}^{i}: i \in \mathcal{A}\}$. Note that since the preambles for each user is independently designed using an i.i.d $Bern(q)$ random variable, our equivalent channel model has the additional constraint that the only tunable parameter available for ``codebook design" is the sampling probability $q$. Hence, identifying active devices reduces to the problem of decoding messages in this equivalent constrained channel. Clearly, the minimum user identification cost for the non-coherent $(\ell,k)-$MnAC is at least the number of channel-uses required to communicate the set of active users; i.e., $\log{\ell\choose k}$ bits of information over this equivalent constrained  point-to-point channel. We formally setup this equivalent channel model below.

Considering the channel in Fig. 1, since the set of active users is $\mathcal{A}=\{a_1,\ldots a_k\}$, the input to the non-coherent $(\ell,k)-$MnAC  is $\textbf{X}=(X^{a_1},\ldots X^{a_k})$.   Thus,  the signal at the input of envelope detector is 
    $ S=\sqrt{P}\sum_{m=1}^{k}h^{a_m} +w.$
 Let $V=\frac{\sum_{i=1}^{k}X^{a_i}}{\sqrt{P}}$ denote the Hamming weight of $\textbf{X}$ which is the number of active users transmitting `On' signal during the $t^{th}$  channel-use.
Thus, conditioned on $V=v$, $U:=|S|^{2}$ follows an exponential  distribution given by

\begin{equation}
     f_{U|V}(u|v)=\frac{1}{v \sigma^{2}P+\sigma_{w}^{2}} e^{-\frac{u}{v \sigma^{2}P+\sigma_{w}^{2}}}, u \geq 0.
     \label{expdis}
 \end{equation} 
 As evident from (\ref{threshold}) and (\ref{expdis}), we have $\textbf{X} \rightarrow V\rightarrow Z$, i.e.,   the transition probability $p\left(z\mid \textbf{x},v\right)$  is only dependent on the  channel input  $\textbf{X}=(X^{a_1},X^{a_2},\ldots X^{a_k})$  through its Hamming weight  $V$.  Also,  since $  f_{U|V}$ is an exponential distribution, $p_v:=p\left(Z=0\mid V=v\right) = p\left(U <\gamma\mid V=v\right)$ can be expressed as
\begin{equation}
p_v=1-e^{-\frac{\gamma}{v \sigma^{2}P+\sigma_{w}^{2}}}. \label{channeleq}
\end{equation} Similarly, $\operatorname{Pr}\left(Z=1\mid V=v\right)=1-p_v.$ Note that $p_v$ is a strictly decreasing function of $v$ where $v \in\{0,1,..,k\}$ assuming  w.l.o.g. positive values for $\sigma^2, \sigma_w^2$ and $\gamma$.

Thus, the non-coherent $(\ell,k)-$MnAC can be equivalently viewed as a traditional point-to-point communication channel whose inputs are the active users as in Fig. \ref{fig:redalpha}. Moreover, this equivalent channel can be modeled as a cascade of two channels; the first channel computes the Hamming weight $V$ of the input $\textbf{X}$ whereas the second channel translates the Hamming weight $V$ into a binary output $Z$ depending on the fading statistics $(\sigma^2)$, noise variance $(\sigma_w^2)$ of the wireless channel and the non-coherent detector threshold $\gamma$. We exploit this cascade channel structure to establish the minimum user identification cost $n(\ell)$ for the non-coherent $(\ell,k)-$MnAC.

\begin{figure}   
	\centering
	\resizebox{3.3in}{!}{

\tikzset{every picture/.style={line width=0.85pt}} 

\begin{tikzpicture}[x=0.85pt,y=0.85pt,yscale=-.85,xscale=.96]

\draw   (59.92,86.98) -- (92.32,86.98) -- (92.32,119.45) -- (59.92,119.45) -- cycle ;
\draw   (59.92,131.79) -- (92.32,131.79) -- (92.32,167.63) -- (59.92,167.63) -- cycle ;
\draw   (59.92,221.39) -- (92.32,221.39) -- (92.32,257.23) -- (59.92,257.23) -- cycle ;
\draw  [dash pattern={on 1.84pt off 2.51pt}]  (76.25,175.25) -- (76.77,216.83) ;
\draw    (92.82,145.23) -- (145.76,145.95) ;
\draw [shift={(147.76,145.97)}, rotate = 180.78] [color={rgb, 255:red, 0; green, 0; blue, 0 }  ][line width=0.75]    (10.93,-3.29) .. controls (6.95,-1.4) and (3.31,-0.3) .. (0,0) .. controls (3.31,0.3) and (6.95,1.4) .. (10.93,3.29)   ;
\draw    (91.83,239.31) -- (144.78,240.03) ;
\draw [shift={(146.78,240.06)}, rotate = 180.78] [color={rgb, 255:red, 0; green, 0; blue, 0 }  ][line width=0.75]    (10.93,-3.29) .. controls (6.95,-1.4) and (3.31,-0.3) .. (0,0) .. controls (3.31,0.3) and (6.95,1.4) .. (10.93,3.29)   ;
\draw    (92.82,107.59) -- (145.76,108.31) ;
\draw [shift={(147.76,108.34)}, rotate = 180.78] [color={rgb, 255:red, 0; green, 0; blue, 0 }  ][line width=0.75]    (10.93,-3.29) .. controls (6.95,-1.4) and (3.31,-0.3) .. (0,0) .. controls (3.31,0.3) and (6.95,1.4) .. (10.93,3.29)   ;
\draw   (155.99,80.56) -- (262.26,80.56) -- (262.26,259.77) -- (155.99,259.77) -- cycle ;
\draw    (263.57,96.84) -- (285.26,96.7) ;
\draw [shift={(287.26,96.69)}, rotate = 179.64] [color={rgb, 255:red, 0; green, 0; blue, 0 }  ][line width=0.75]    (10.93,-3.29) .. controls (6.95,-1.4) and (3.31,-0.3) .. (0,0) .. controls (3.31,0.3) and (6.95,1.4) .. (10.93,3.29)   ;
\draw    (263.57,116.55) -- (285.26,116.42) ;
\draw [shift={(287.26,116.4)}, rotate = 179.64] [color={rgb, 255:red, 0; green, 0; blue, 0 }  ][line width=0.75]    (10.93,-3.29) .. controls (6.95,-1.4) and (3.31,-0.3) .. (0,0) .. controls (3.31,0.3) and (6.95,1.4) .. (10.93,3.29)   ;
\draw    (263.57,250.06) -- (285.26,249.92) ;
\draw [shift={(287.26,249.91)}, rotate = 179.64] [color={rgb, 255:red, 0; green, 0; blue, 0 }  ][line width=0.75]    (10.93,-3.29) .. controls (6.95,-1.4) and (3.31,-0.3) .. (0,0) .. controls (3.31,0.3) and (6.95,1.4) .. (10.93,3.29)   ;
\draw  [dash pattern={on 0.84pt off 2.51pt}]  (310.95,137.91) -- (310.95,147.77) -- (310.95,158.52) ;
\draw  [dash pattern={on 0.84pt off 2.51pt}]  (280.95,50.91) -- (280.95,248.52) ;
\draw   (468.87,104.91) .. controls (468.87,92.53) and (479.92,82.5) .. (493.55,82.5) .. controls (507.18,82.5) and (518.22,92.53) .. (518.22,104.91) .. controls (518.22,117.28) and (507.18,127.31) .. (493.55,127.31) .. controls (479.92,127.31) and (468.87,117.28) .. (468.87,104.91) -- cycle ;
\draw   (470.85,235.73) .. controls (470.85,223.35) and (481.89,213.32) .. (495.52,213.32) .. controls (509.15,213.32) and (520.2,223.35) .. (520.2,235.73) .. controls (520.2,248.1) and (509.15,258.13) .. (495.52,258.13) .. controls (481.89,258.13) and (470.85,248.1) .. (470.85,235.73) -- cycle ;
\draw    (321.81,91.32) -- (473.41,249.36) ;
\draw [shift={(474.79,250.81)}, rotate = 226.19] [color={rgb, 255:red, 0; green, 0; blue, 0 }  ][line width=0.75]    (10.93,-3.29) .. controls (6.95,-1.4) and (3.31,-0.3) .. (0,0) .. controls (3.31,0.3) and (6.95,1.4) .. (10.93,3.29)   ;
\draw    (322.79,118.2) -- (473.29,249.49) ;
\draw [shift={(474.79,250.81)}, rotate = 221.1] [color={rgb, 255:red, 0; green, 0; blue, 0 }  ][line width=0.75]    (10.93,-3.29) .. controls (6.95,-1.4) and (3.31,-0.3) .. (0,0) .. controls (3.31,0.3) and (6.95,1.4) .. (10.93,3.29)   ;
\draw    (324.77,250.81) -- (472.79,250.81) ;
\draw [shift={(474.79,250.81)}, rotate = 180] [color={rgb, 255:red, 0; green, 0; blue, 0 }  ][line width=0.75]    (10.93,-3.29) .. controls (6.95,-1.4) and (3.31,-0.3) .. (0,0) .. controls (3.31,0.3) and (6.95,1.4) .. (10.93,3.29)   ;
\draw    (322.79,118.2) -- (472.82,91.66) ;
\draw [shift={(474.79,91.32)}, rotate = 169.97] [color={rgb, 255:red, 0; green, 0; blue, 0 }  ][line width=0.75]    (10.93,-3.29) .. controls (6.95,-1.4) and (3.31,-0.3) .. (0,0) .. controls (3.31,0.3) and (6.95,1.4) .. (10.93,3.29)   ;
\draw    (321.81,91.32) -- (472.79,91.32) ;
\draw [shift={(474.79,91.32)}, rotate = 180] [color={rgb, 255:red, 0; green, 0; blue, 0 }  ][line width=0.75]    (10.93,-3.29) .. controls (6.95,-1.4) and (3.31,-0.3) .. (0,0) .. controls (3.31,0.3) and (6.95,1.4) .. (10.93,3.29)   ;
\draw    (324.77,250.81) -- (473.42,92.77) ;
\draw [shift={(474.79,91.32)}, rotate = 133.25] [color={rgb, 255:red, 0; green, 0; blue, 0 }  ][line width=0.75]    (10.93,-3.29) .. controls (6.95,-1.4) and (3.31,-0.3) .. (0,0) .. controls (3.31,0.3) and (6.95,1.4) .. (10.93,3.29)   ;
\draw   (301.08,80.56) -- (319.83,80.56) -- (319.83,259.77) -- (301.08,259.77) -- cycle ;
\draw  [dash pattern={on 0.84pt off 2.51pt}]  (310.95,200.63) -- (310.95,210.49) -- (309.96,233.78) ;
\draw    (322.79,171.96) -- (473.03,92.25) ;
\draw [shift={(474.79,91.32)}, rotate = 152.05] [color={rgb, 255:red, 0; green, 0; blue, 0 }  ][line width=0.75]    (10.93,-3.29) .. controls (6.95,-1.4) and (3.31,-0.3) .. (0,0) .. controls (3.31,0.3) and (6.95,1.4) .. (10.93,3.29)   ;
\draw    (322.79,171.96) -- (473.02,249.89) ;
\draw [shift={(474.79,250.81)}, rotate = 207.42] [color={rgb, 255:red, 0; green, 0; blue, 0 }  ][line width=0.75]    (10.93,-3.29) .. controls (6.95,-1.4) and (3.31,-0.3) .. (0,0) .. controls (3.31,0.3) and (6.95,1.4) .. (10.93,3.29)   ;
\draw  [fill={rgb, 255:red, 155; green, 155; blue, 155 }  ,fill opacity=0.23 ] (119.47,103.11) .. controls (119.47,79.36) and (138.72,60.1) .. (162.48,60.1) -- (402.17,60.1) .. controls (425.93,60.1) and (445.18,79.36) .. (445.18,103.11) -- (445.18,232.14) .. controls (445.18,255.89) and (425.93,275.15) .. (402.17,275.15) -- (162.48,275.15) .. controls (138.72,275.15) and (119.47,255.89) .. (119.47,232.14) -- cycle ;

\draw (165,94.4) node [anchor=north west][inner sep=0.75pt]  [font=\huge]  {$\frac{\sum _{i=1}^{k}X_{t}^{a_{i}}}{\sqrt{P}}$};
\draw (156,173) node [anchor=north west][inner sep=0.75pt]  [font=\large] [align=left] {{\large \ \ \ Hamming  }\\{\large \ \ \ weight $(V)$}\\{\large \  computation}};
\draw (28,32) node [anchor=north west][inner sep=0.75pt]  [font=\small] [align=left] {{\large Channel input}\\{\large  $\displaystyle \ \ \ \ \ \ \ \ \textbf{X}$}};
\draw (220.53,30.96) node [anchor=north west][inner sep=0.75pt]  [font=\large]  {{Hamming weight}{$\hspace{0.1cm}V$}};
\draw (487.53,95.96) node [anchor=north west][inner sep=0.75pt]  [font=\LARGE]  {$0$};
\draw (488.53,225.39) node [anchor=north west][inner sep=0.75pt]  [font=\LARGE]  {$1$};
\draw (421,32) node [anchor=north west][inner sep=0.75pt]  [font=\large] [align=left] {{\large Channel Output}\\{\large $\displaystyle \ \ \ \ \ \ \ \ \ \ Z$}};
\draw (48.2,276.73) node [anchor=north west][inner sep=0.75pt]  [font=\large]  {$active\ user\ set:\ \{a_{1} ,a_{2} ,\dotsc ,a_{k} \}$};
\draw (306.92,168.09) node [anchor=north west][inner sep=0.75pt]    {$i$};
\draw (328.51,150.22) node [anchor=north west][inner sep=0.75pt]  [font=\large,rotate=-332.77]  {$p_{i}$};
\draw (329.93,176.77) node [anchor=north west][inner sep=0.75pt]  [font=\large,rotate=-29.2]  {$1-p_{i}$};
\draw (49.02,302.55) node [anchor=north west][inner sep=0.75pt]  [font=\large]  {$X_{t}^{a_{i}} \in \{0,\sqrt{P}\}$};
\draw (348.88,282.21) node [anchor=north west][inner sep=0.75pt]  [font=\Large]  {$p_{i} =1-e^{-\frac{\gamma}{i\sigma ^{2}P +\sigma _{w}^{2}}}$};
\draw (61.92,90.38) node [anchor=north west][inner sep=0.75pt]  [font=\large]  {$X_{t}^{a_{1}}$};
\draw (60.03,226.07) node [anchor=north west][inner sep=0.75pt]  [font=\large]  {$X_{t}^{a_{k}}$};
\draw (60.74,135.53) node [anchor=north west][inner sep=0.75pt]  [font=\large]  {$X_{t}^{a_{2}}$};
\draw (306.92,240.67) node [anchor=north west][inner sep=0.75pt]    {$k$};
\draw (305.94,110.75) node [anchor=north west][inner sep=0.75pt]    {$1$};
\draw (305.94,88.35) node [anchor=north west][inner sep=0.75pt]    {$0$};
\end{tikzpicture}
}

	\setlength{\belowcaptionskip}{-19pt}
	\caption{Equivalent channel with only  active users of the $(\ell,k)-$MnAC as inputs.}
	\label{fig:redalpha}
	\end{figure}
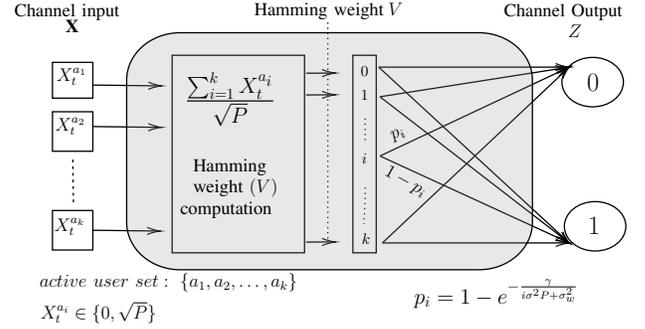

\subsection{ Maximum rate of the equivalent channel}
   
\begin{lemma}
For a  given fading statistics $\sigma^2$ and noise variance $\sigma_w^2$ for the $(\ell,k)$-MnAC, the maximum rate of the equivalent point-to-point channel is \begin{equation} 
    C=\max_{(\gamma,q)}h\Big(E\Big[e^{-\frac{\gamma}{V  \sigma^{2}P+\sigma_{w}^{2}}}\Big]\Big)-E\Big[h\Big(e^{-\frac{\gamma}{V  \sigma^{2}P+\sigma_{w}^{2}}}\Big)\Big]  \label{jengap}
\end{equation} 
 where $E(\cdot)$ denotes expectation w.r.t  $ V \sim Bin(k, q)$ and $h(x)=-x \log x-(1-x) \log (1-x)$ is the binary entropy. Here, the maximization is done jointly over the sampling probability $q$ for preamble generation and  the energy detector threshold $\gamma$.
 \label{thm2th}
 \end{lemma}
 \vspace{-.9cm}
\begin{proof} \textit{(Sketch)} Our equivalent point-to-point communication channel described above has two tunable parameters, viz, the sampling probability $q$ at the user side and the non-coherent detector threshold $\gamma$ at the BS. Thus, the maximum rate of this equivalent channel between the binary vector input $\textbf{X}$ and binary scalar output $Z$ is  $C=\max_{(\gamma,q)} I(\textbf{X};Z)$.
   Since the preamble entries $X^{a_i}  \sim Bern(q)$  such that  $X^{a_i}\ind X^{a_j}  $ $\forall i,j \in  \{1,\ldots k\}; i\neq j$, the Hamming weight $V$ follows Binomial distribution denoted by $V \sim  Bin(k,q).$ $V$ being a deterministic function of $\textbf{X}$ and $\textbf{X} \rightarrow V\rightarrow Z$, we have
    \begin{equation}
    \begin{aligned}
        C &=\max_{(\gamma,q)} I(V;Z)\\ &=\max_{(\gamma,q)} H(Z)-H(Z \mid V) \text{ where } V \sim Bin(k, q). 
    \end{aligned}
\label{eq20}
\end{equation} Using standard algebraic manipulations, we can write \small{$H(Z)=h\Big(E\Big[e^{-\frac{\gamma}{V  \sigma^{2}P+\sigma_{w}^{2}}}\Big]\Big)$}  and \small{$H(Z \mid V) = E\left[h\Big(e^{-\frac{\gamma}{V  \sigma^{2}P+\sigma_{w}^{2}}}\Big)\right]$}.
Thus, jointly optimizing over $\gamma$ and $q$, the maximum rate is given by 
\begin{equation*}
    C=\max_{(\gamma,q)}h\Big(E\Big[e^{-\frac{\gamma}{V  \sigma^{2}P+\sigma_{w}^{2}}}\Big]\Big)-E\Big[h\Big(e^{-\frac{\gamma}{V  \sigma^{2}P+\sigma_{w}^{2}}}\Big)\Big] ,
\end{equation*}  completing the proof. \end{proof}	
Note that the minimum user identification  cost $n(\ell)$ for the non-coherent $(\ell,k)-$MnAC,  must obey $ C \geq \frac{\log {\ell\choose k}}{n(\ell)}$. However,  we cannot directly  claim that  $n(\ell)$  exactly matches the number of channel-uses required to send $\log{\ell\choose k}$ bits of information over the equivalent channel at maximum rate $C$. The reason is that the effective  codeword corresponding to a user set $\mathcal{A}$ is the subset of preambles $\{\textbf{X}^{i}: i \in \mathcal{A}\}$ indexed by the set $\mathcal{A}$.  Thus, the effective codewords of different user sets are potentially dependent due to  overlapping preambles. This is in contrast to the independent codeword assumptions typical to channel capacity analysis thereby rendering the traditional random coding approach ineffective.  We close this gap by making use of the equivalence of active device identification and GT to propose a decoding scheme achieving the minimum user identification cost as presented below.
\subsection{Minimum user identification cost}
\begin{theorem}
The minimum  user identification cost $n(\ell)$ of the non-coherent  $(\ell,k)-MnAC$ with $k=\Theta(1)$  is given by
\begin{equation*} \small
    n(\ell)= \frac{k\log(\ell)}{\max_{(\gamma,q)}h\Big(E\Big[e^{-\frac{\gamma}{V  \sigma^{2}P+\sigma_{w}^{2}}}\Big]\Big)-E\Big[h\Big(e^{-\frac{\gamma}{V  \sigma^{2}P+\sigma_{w}^{2}}}\Big)\Big]}
\end{equation*} \normalsize where $E(.)$ denotes the expectation w.r.t  $V \sim \operatorname{Bin}(k, q)$, $q$ is the sampling probability for i.i.d preamble generation and $\gamma$ denotes the threshold for non-coherent energy detection.
\label{thm3}
\end{theorem}
\begin{proof} \textit{(Sketch)} Using standard capacity arguments\cite{10.5555/1146355}  along with Lemma $\ref{thm2th}$, it is evident that user identification cost in the non-coherent $(\ell,k)$-MnAC must obey $ n(\ell) \geq \frac{\log {\ell\choose k}}{C}$ where $C$ is as given in Lemma \ref{thm2th}.
Thus, in the $k =\Theta(1)$ regime we have
\footnotesize{\begin{equation} 
    n(\ell) \geq \frac{k\log(\ell)}{\max_{(\gamma,q)}h\Big(E\Big[e^{-\frac{\gamma}{V  \sigma^{2}P+\sigma_{w}^{2}}}\Big]\Big)-E\Big[h\Big(e^{-\frac{\gamma}{V  \sigma^{2}P+\sigma_{w}^{2}}}\Big)\Big]}.\label{equ}
\end{equation}} \normalsize
 
 Next, we present a decoder as shown in Algorithm \ref{alg:MYALG}  which can  achieve the equality  in $(\ref{equ})$ relying on classical thresholding techniques in information theory \cite{1057459}.  Due to the equivalence of active user identification and GT, we use an  adaptation of  the GT decoding strategy  in  \cite{6157065} to the context of user identification in  non-coherent $(\ell,k)$-MnAC.
\begin{algorithm}
\caption{Decoder achieving the minimum user identification cost in the $k=\Theta(1)$ regime. }\label{alg:cap}
 
\begin{algorithmic}[1]
\State \textbf{Given}: Preamble set, $\mathcal{P} \triangleq \{\textbf{X}^i: i\in \mathcal{D}\}$;
Detector outputs: $\textbf{Z}\triangleq(Z_1,\ldots, Z_n)$; $k$: Number of active devices.
\State \textbf{Initialize}: Set of all size-$k$ subsets of $\mathcal{D}$: $\mathcal{S} =\{\mathcal{K}_1,\mathcal{K}_2,\ldots,\mathcal{K}_{{\ell \choose k}}\}$;Estimate of activity status: $\hat{\boldsymbol{\beta}}=\mathbf{0}_{1 \times \ell};$  $\mathcal{A}_{est}=\Phi.$
\For{$i =1: \left|\mathcal{S}\right|$ }
\State  $\tilde{\mathcal{A}}= \mathcal{K}_i$; Number of threshold tests passed, $n_t=0$.
\For{$j =1:k$}
\State \small \hspace{-.3cm}$\Lambda^{\{j\}}=\left\{\left(\mathcal{I}^0, \mathcal{I}^1\right): \mathcal{I}^0 \bigcap \mathcal{I}^1=\emptyset, \mathcal{I}^0 \bigcup \mathcal{I}^1=\tilde{\mathcal{A}},\left|\mathcal{I}^0\right|=j\right\}$ 
\State $\eta_{j}=\log \frac{\rho}{k{k \choose j}{\ell-k \choose j}}$ for some $\rho >0$.
\If   {$\log \left( \frac{ p\left(\textbf{z} \mid \textbf{x}({\mathcal{I}^0}),\textbf{x}({\mathcal{I}^1})\right)}{ p\left(\textbf{z} \mid \textbf{x}({\mathcal{I}^1})\right)}\right) >\eta_{j}, \forall (\mathcal{I}^0,\mathcal{I}^1) \in \Lambda^{\{j\}}$}\normalsize
\State $n_t=n_t+{k \choose j}$.
\EndIf
\EndFor
\State \textbf{if }{$n_t =2^k-1,$}
$\mathcal{A}_{est}=\mathcal{A}_{est} \bigcup \tilde{\mathcal{A}}. $
\textbf{end if}
\EndFor
\If {$\left|\mathcal{A}_{est}\right|= k$},
  $\hat{\boldsymbol{\beta}}:\hat{{\beta_i}}=1, \forall i \in \mathcal{A}_{est}$, \textbf{else} \textbf{ error.} 
 \EndIf
\end{algorithmic}
\label{alg:MYALG} \vspace{0cm}
\end{algorithm} 
The algorithm operates as follows: Given the set of preambles $\mathcal{P}$ and the detector outputs $\textbf{Z}$, the decoder exhaustively searches for an approximate maximum likelihood (ML) size-$k$ set from  all possible ${\ell \choose k}$ subsets. The decoder declares $\mathcal{A}^*$ as the active  set iff
$
    p\left(\textbf{z} \mid \textbf{x}({\mathcal{A}^*})\right)>p\left(\textbf{z} \mid \textbf{x}({\tilde{\mathcal{A}} })\right) ;$  $\forall \tilde{\mathcal{A}}  \neq \mathcal{A}^*,$ where $\textbf{x}(\mathbf{\mathcal{A}})$ denotes $\{\textbf{x}^{a_i}:a_i\in \mathcal{A}\}$.  The corresponding activity status vector is  $\boldsymbol{\beta}^{*}:{{\beta_i}}=1, \forall i \in \mathcal{A}^*$. We perform this search  through a series of threshold tests for each candidate set $\tilde{\mathcal{A}} \in \mathcal{S}$. For each $\tilde{\mathcal{A}}$, we consider a partition $( \mathcal{I}^0, \mathcal{I}^1)$  of the candidate active set into its overlap and difference with the true active set as in step (6). i.e., $ \mathcal{I}^1:=\mathcal{A} \cap  \tilde{\mathcal{A}} $ and $ \mathcal{I}^0:=\tilde{\mathcal{A}} \setminus {\mathcal{A}} $. Then, we perform $2^k-1$ threshold tests as below. 
\begin{equation} \small
    \log \left( \frac{ p\left(\textbf{z} \mid \textbf{x}( \mathcal{I}^0),\textbf{x}( \mathcal{I}^1)\right)}{ p\left(\textbf{z} \mid \textbf{x}( \mathcal{I}^1)\right)}\right) >\eta_{s}, \forall ( \mathcal{I}^0, \mathcal{I}^1) \text{ partitions of } \tilde{\mathcal{A}}, \label{deco} 
\end{equation} \normalsize
for predetermined thresholds $\eta_{s}$ where $s:=| \mathcal{I}^0| \in \{1,2,\ldots,k\}$. We declare $\tilde{\mathcal{A}}$ as the  active user set if it is the only set passing all threshold tests. An error occurs if the active user set fails at least one of the tests in  (\ref{deco}) or an incorrect set passes all the tests.  
As shown in \cite{6157065}, using a set of thresholds given by $ \eta_{s}=\log \frac{\rho}{k{k \choose s}{\ell-k \choose s}}$,
 we can upper bound $\mathbb{P}_e^{(\ell)}$ as \footnotesize{\begin{equation*} 
    \mathbb{P}_e^{(\ell)} \leq \sum_{s=1}^{k}{k \choose s} \mathbb{P}\Bigg( \log \left( \frac{ p\left(\textbf{z} \mid \textbf{x}( \mathcal{I}^0),\textbf{x}( \mathcal{I}^1)\right)}{ p\left(\textbf{z} \mid \textbf{x}( \mathcal{I}^1)\right)}\right)\leq \log  \frac{\rho}{k{k \choose s}{\ell-k \choose s}}\Bigg)+\rho. \label{nonasym}
\end{equation*}} \normalsize

Since  preamble design uses i.i.d Bernoulli random variables, using concentration inequalities for the $k=\Theta(1)$ regime, we can achieve $\mathbb{P}_e^{(\ell)}\rightarrow 0,$ with $\ell \rightarrow \infty$ as long as $n(\ell)$ obeys
\begin{equation}
     n(\ell) \geq \max _{s=1, \ldots, k} \frac{\log{\ell-k \choose s}}{I(\textbf{X}({ \mathcal{I}^0}_s);Z \mid \textbf{X}({ \mathcal{I}^1}_s) )}(1+o(1)).\label{teq}
\end{equation} Here, $\textbf{X}({ \mathcal{I}^0}_s):=\{\textbf{X}({ \mathcal{I}^0}):| \mathcal{I}^0|=s\}$ and $\textbf{X}({ \mathcal{I}^1}_s):=\{\textbf{X}({ \mathcal{I}^1}):| \mathcal{I}^0|=s\}$.
Using the asymptotic expression
$
\log {\ell -k \choose s}=s \log \ell+O(1)
$ in the $k=\Theta(1)$ regime implies 
\begin{equation}
     n(\ell) \geq \max _{s=1, \ldots, k} \frac{s\log \ell}{I(\textbf{X}({ \mathcal{I}^0}_s);Z \mid \textbf{X}({ \mathcal{I}^1}_s) )}(1+o(1)).\label{teq1}
\end{equation}
We can rewrite $I_s:=I(\textbf{X}({ \mathcal{I}^0}_s);Z \mid \textbf{X}({ \mathcal{I}^1}_s) )$   as 
\begin{equation} \small
    I_s=\sum_{j=k-s+1}^{k}\Big(H\big(\textbf{X}(j)\big)-H\big(\textbf{X}(j) \mid Z, \mathbf{X}(1,2,..,j-1)\big)\Big) \label{con}.
\end{equation}Note that $\frac{I_s}{s}$ attains its minimum when $s=k$ since conditioning reduces entropy \cite{10.5555/1146355}. Thus, when $k=\Theta(1)$, (\ref{teq1}) simplifies to 
$ n(\ell) \geq \frac{k\log(\ell)}{I(\textbf{X};Z  )}(1+o(1)).$  Using the maximum rate of the equivalent channel $C$  in Lemma 1, we can conclude that,  our proposed decoder has $\mathbb{P}_e^{(\ell)} \rightarrow 0$ with $\ell \rightarrow \infty$ as long as 
 \footnotesize{\begin{equation*}
    n(\ell) \geq \frac{k\log(\ell)}{\max_{(\gamma,q)}h\Big(E\Big[e^{-\frac{\gamma}{V  \sigma^{2}P+\sigma_{w}^{2}}}\Big]\Big)-E\Big[h\Big(e^{-\frac{\gamma}{V  \sigma^{2}P+\sigma_{w}^{2}}}\Big)\Big]}(1+o(1)),
\end{equation*}} \normalsize
completing the proof. 
\end{proof}
\section{Belief Propagation with Soft Thresholding}

The decoder in Algorithm \ref{alg:MYALG} achieving the minimum user identification cost is impractical as it requires an exhaustive search over all candidate active device sets which is computationally infeasible. In this section, we present a practical scheme for user identification based on BP which is well-known to have excellent performance in compressed sensing problems including LDPC decoding over noisy channels\cite{1495850}. While BP has been extensively studied in the context of  sparse recovery \cite{8849288,https://doi.org/10.48550/arxiv.1102.3289}, to the best of our knowledge, it was not used for  user identification in non-coherent fast fading MnAC using On-Off signaling. Although, BP  has been proposed for noisy GT decoding as a practical approach  \cite{sejdinovic2010note}, this was for a simple channel model involving the addition and dilution noise models. In contrast, our noise model in $(\ref{channeleq})$ incorporating channel fading  effects and non-coherent detection at the BS side is more intricate. Due to the equivalence between active device identification and GT, we use  a similar BP approach for our active device identification problem in  non-coherent MnAC since the distribution of the binary energy detector outputs after each channel-use is only dependent on the number of active devices involved in that channel-use and is not dependent on their specific identities.

Given the preamble set $\mathcal{P}$ and  binary detector outputs $\textbf{Z}=(Z_1,\ldots, Z_n)$, BP relies on approximate MAP decoding  using the marginals of the posterior distribution to identify the activity status vector $\boldsymbol{\beta}$\cite{sejdinovic2010note}. i.e.,
	   $ \widehat{\beta_i}=\underset{\beta_i \in\{0,1\}}{\arg\max }\hspace{2pt} p\left(\beta_i \mid \mathbf{Z}\right).$   It can be rewritten as 
\begin{equation}
 \hat{\beta}_i= \underset{\beta_i \in\{0,1\}}{\arg\max } \sum_{\sim \beta_i}\left[\prod_{t=1}^n p\left(Z_t \mid \boldsymbol{\beta}\right) \prod_{j=1}^\ell p\left(\beta_j\right)\right].
 \label{bpexp} 
\end{equation} Eqn. (\ref{bpexp}) can be approximately computed using loopy belief propagation based on a bipartite graph with $\ell$ devices at one side and the $n$ energy detector  outputs at the other side. The priors $p\left(\beta_j\right)$'s are initialized to $\frac{k}{\ell}$ and updated through message passing in further iterations.
Note that $p(Z_t \mid \boldsymbol{\beta})=p\left(Z_t \mid \sum\limits_{{i\in \mathcal{A}; X_t^i =\sqrt{P}}}\beta_{i}\right)$, since $Z_t$ depends on $\boldsymbol{\beta}$ only through  the number of active devices transmitting `On' symbols during the $t^{th}$ channel-use. This symmetry due to our i.i.d fast fading random access channel model leads to efficient computation of BP message update rules \cite{sejdinovic2010note}. We adopt BP with Soft Thresholding (BP-ST) strategy in which the $k$ users with the highest marginals are classified as active.  

We conducted simulations to evaluate the performance of BP-ST and its gap w.r.t. the minimum user identification cost. We considered $\ell =1000$ devices out of which $k=25$ are active. We used 20 message passing iterations. The SNR was varied from  0 to 10 dB by tuning the fading statistics $\sigma^2$ while keeping $P$ and $\sigma_w^2$ fixed at 0 dB. For a given SNR, we noted that other choices of the parameters $P, \sigma^2$ and  $\sigma_w^2$ lead to similar performance. Fig. \ref{bpst} shows that as the SNR  increases, devices are identified with fewer channel-uses. In Fig. \ref{fig:z4}, we compare the minimum user identification cost in Theorem 1 with the performance of BP-ST. Evidently, the gap between the BP-ST curve and the minimum user identification cost reduces notably in the high SNR regime. We also observed that (excluded for brevity)  BP can approach the minimum user identification cost if we relax the success criterion slightly by allowing room for few misdetections.
\begin{figure}   
  \centering
	\begin{tikzpicture}
  \sbox0{\includegraphics[width=.8\linewidth,height=53mm,trim={1.4cm 0.8cm 0cm 0},clip]{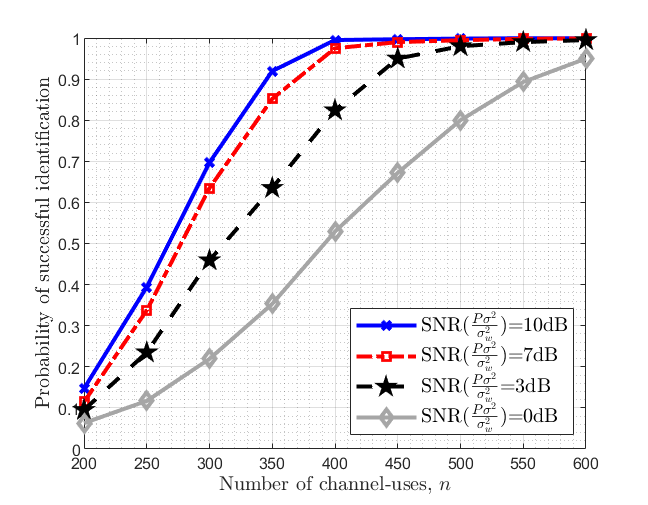}}
  \node[above right,inner sep=0pt] at (0,0)  {\usebox{0}};
  \node[black] at (0.5\wd0,-0.06\ht0) {\footnotesize{Number of channel-uses, $n$}};
   \node[black] at (0.75\wd0,0.186\ht0) {\scriptsize{)}};
  \node[black,rotate=90] at (-0.04\wd0,0.5\ht0) {\footnotesize{$\mathbb{P}$(success)}};
\end{tikzpicture}
	\setlength\abovecaptionskip{0pt}
	\setlength{\belowcaptionskip}{-15pt} 
	\caption{BP-ST for $(1000,25)$-MnAC.} 
	\label{bpst}
	\end{figure}
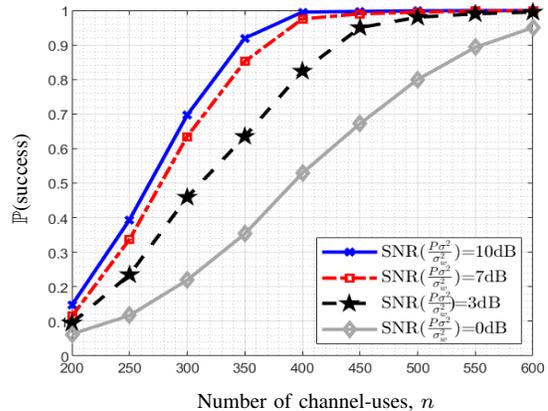

\begin{figure}   
	\centering
	\begin{tikzpicture}
  \sbox0{\includegraphics[width=.8\linewidth,height=53mm,trim={1.2cm 0.8cm 0 0},clip]{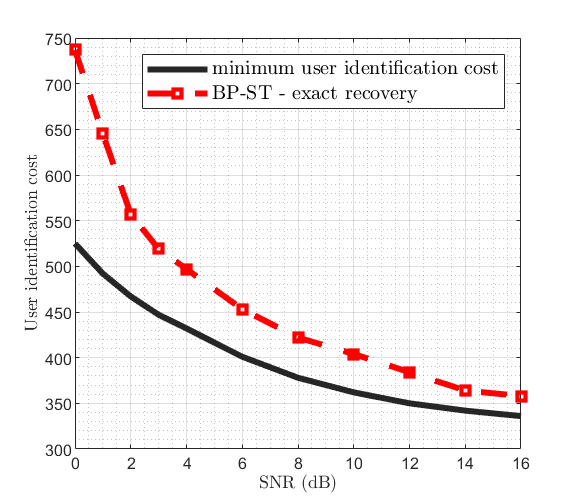}}
  \node[above right,inner sep=0pt] at (0,0)  {\usebox{0}};
  \node[black] at (0.5\wd0,-0.06\ht0) {\footnotesize{SNR (dB)}};
  \node[black,rotate=90] at (-0.04\wd0,0.5\ht0) {\footnotesize{ Number of channel-uses}};
\end{tikzpicture}
	\setlength\abovecaptionskip{0pt}
	\setlength{\belowcaptionskip}{-15pt} 
	\caption{ Minimum user identification cost versus BP-ST for $(1000,25)$-MnAC. }
	\label{fig:z4}
	\end{figure}
	\section{Conclusion}
	In this paper, we have proposed an active user discovery framework for massive random access based on On-Off preamble transmission and non-coherent energy detection at the receiver. We have evaluated the performance of our proposed scheme  in terms of the minimum number of channel-uses required for reliable active user identification. We have established the  minimum user identification cost and an achievability scheme in the $k=\Theta(1)$ regime. Furthermore, we have proposed BP-ST, a practical scheme for device identification based on belief propagation.

\bibliographystyle{IEEEtran} 
\bibliography{THz_scatter}
\end{document}